\documentclass[12pt]{article}
\usepackage{graphicx}
\usepackage{graphics}
\usepackage{fancybox}
\usepackage{amsmath,float,latexsym,psfrag,epsf,epsfig,amssymb,chicago,rotating}
\usepackage{amsthm,bm}

\usepackage{color}
\usepackage{url}
\usepackage{subcaption}
\usepackage{fullpage}
\usepackage[ruled,linesnumbered]{algorithm2e}
\usepackage{algorithmicx}

\newcommand{\comment}[1]{}

\newtheorem{lemma}{Lemma}[section]             

\newtheorem{case}{Example}

\usepackage{mathtools}

\newcommand{\vect}[1]{\boldsymbol{#1}}

\newcommand{\w}{{\mathbf{w}}}

\newcommand{\y}{{\mathbf{y}}}
\newcommand{\z}{{\mathbf{z}}}

\newcommand{\E}{{\mathbf{E}}}
\newcommand{\V}{{\mathbf{V}}}

\begin{document}

\begin{center}

{\Large \bfseries Investigation of the widely applicable Bayesian information criterion}
\vspace{5 mm}

{\large N. Friel$^\star$\footnote{Address for correspondence: \texttt{nial.friel@ucd.ie}}, J.P. McKeone$^\dagger$$^{\star\star}$, C.J. Oates$^\ddagger$$^{\star\star}$, A.N. Pettitt$^\dagger$$^{\star\star}$. } \\
{\textit{$^\star$School of Mathematics and Statistics and Insight Centre for Data Analytics, \\ University College Dublin, Ireland.\\
$^\ddagger$School of Mathematical and Physical Sciences, University of Technology, Sydney.\\
$^\dagger$School of Mathematical Sciences, Queensland University of Technology, Brisbane, Queensland, Australia.\\
$^{\star\star}$Australian Research Council Centre for Excellence in Mathematical and Statistical Frontiers. }} 

\today 

\vspace{5mm}

\end{center}

\begin{abstract}
The widely applicable Bayesian information criterion (WBIC) is a simple and fast approximation to the model evidence that has received little practical 
consideration. WBIC uses the fact that the log evidence can be written as an expectation, with respect to a powered posterior proportional to the likelihood raised
to a power $t^*\in{(0,1)}$, of the log deviance. Finding this temperature value $t^*$ is generally an intractable problem. 
We find that for a particular tractable statistical model that the mean squared error of an optimally-tuned version of WBIC with correct temperature $t^*$
is lower than an optimally-tuned version of thermodynamic integration (power posteriors). However in practice WBIC uses the a canonical choice of
$t=1/\log(n)$. Here we investigate the performance of WBIC in practice, for a range of statistical models, both regular models and singular models 
such as latent variable models or those with a hierarchical structure for which BIC cannot provide an adequate solution. Our findings are that, generally WBIC performs 
adequately when one uses informative priors, but it can systematically overestimate the evidence, particularly for small sample sizes. 

\end{abstract}

\bibliographystyle{mybib}

\section{Introduction}

The Bayesian paradigm offers a principled approach to the issue of model choice, through examination of the model evidence, namely the probability of the data given the model. 
Suppose we are given data $y$ and assume there is a collection of competing models, $m_1,\dots,m_l$, each with associated parameters, $\theta_1,\dots,\theta_l$, respectively. 
Viewing the model indicators as parameters with prior distribution $p(m_k)$, the posterior distribution of interest is 
\begin{equation}
	p(\theta_k,m_k|y) \propto f(y|\theta_k,m_k)p(\theta_k|m_k)p(m_k) \nonumber
\end{equation} 
 where $f(y|\theta_k,m_k)$ is the likelihood of the data under model $m_k$ with parameters $\theta_k$ and $p(\theta_k|m_k)$ is the prior on the parameters in model $m_k$.

The constant of proportionality for the un-normalised posterior distribution for model $m_k$ is the marginal likelihood or evidence, 
\begin{equation}
	p(y|m_k) = \int_{\theta_k} f(y|\theta_k,m_k)p(\theta_k|m_k)d\theta_k. \nonumber
\end{equation}
This is a vital quantity in Bayesian model choice and developing good estimates of it continues to be an active area of research in computational statistics. Henceforth, for brevity of notation,
we will drop the dependence on $m_k$, so that we refer to the evidence, likelihood and prior distribution for a given model as, $p(y), f(y|\theta), p(\theta)$, respectively.

There are a growing number of techniques to evaluate the evidence, see for instance, Gelman and Meng~\citeyear{GelmanMeng1998} for a thorough review of importance, bridge and 
path sampling methods, Robert and Wraith~\citeyear{RobertWraith2009} 
for an updated review of such methods that includes the more recent mixture bridge-sampling approach \cite{ChopinRobert2007}, the generalised harmonic mean estimator 
\cite{GelfandDey1994} and nested sampling (\cite{Skilling2006}, or perhaps, \cite{Burrows1980}), in addition to \cite{FrielWyse2012} who compare the accuracy and 
computational burden of these methods. 

The contribution of this work is to explore a new method of approximating the evidence, the widely applicable Bayesian information criterion (WBIC) of 
Watanabe~\citeyear{Watanabe2013}. WBIC was motivated by the fact that the Bayesian information criterion (BIC or Schwarz criterion) \cite{Schwarz1978} is not applicable to 
singular models. A statistical model is termed regular if the mapping from model parameters to a probability distribution is one-to-one and if its Fisher information matrix 
is positive definite. Otherwise, a statistical model is singular. Singular models arise, for example, in latent variable models such as mixture models, hidden Markov models 
and hierarchical models such as artificial neural networks and so on. Singular models cannot be approximated by a normal distribution, which implies that
BIC and AIC are not appropriate for statistical model choice. Watanabe~\citeyear{Watanabe2013} has shown that WBIC converges to the model evidence, asymptotically as 
$n\rightarrow \infty$, for both singular and regular statistical. In this sense, WBIC is generalisation of BIC. 

WBIC is straightforward to evaluate, requiring only one Markov chain Monte Carlo (MCMC) chain to estimate the evidence. Thus far, WBIC has received no more than a cursory mention by 
Gelman \textit{et al.}~\citeyear{GelmanHwangVehtari2013} and while it has been applied in practice to a specific reduced rank regression model, see unpublished work by Drton and 
Plummer \citeyear{drton:plummer13} and for the case of Gaussian process regression \cite{mononen14}, 
beyond Watanabe's own implementation there has been no further exploration of 
the criterion. The focus of this is to assess its performance as an approximation of the evidence.

The paper is organised as follows, the key results and notation for power posteriors, necessary for WBIC, are presented in Section ~\ref{PPs}. Watanabe's WBIC is presented in 
Section ~\ref{WBIC}. Section~\ref{sec:theory:comparison} presents a theoretical comparison of WBIC and the power posterior approach.
The performance of WBIC compared to several competing methods is illustrated in Section ~\ref{Numerical} for four examples. 
The article is concluded with a brief discussion in Section ~\ref{Discussion}.

\section{Power posteriors}
\label{PPs}

Friel and Pettitt~\citeyear{FrielPettitt2008} propose the method of power posteriors, a path sampling type method, to evaluate the marginal likelihood (or evidence) 
in an application of the thermodynamic integration technique from statistical physics. Dating to Kirkwood~\citeyear{Kirkwood1935}, thermodynamic integration has a 
long history in the statistical physics literature. An in-depth background to thermodynamic integration and Bayes free energy
(also known as marginal likelihood) calculations 
for context specific statistical models is given by Chiput and Pohorille~\citeyear{ChiputPohorille2007}. In addition, the slow growth method of 
Bash \textit{et al.}~\citeyear{BashSinghLangridgeKollman1987} is a notable forerunner to the method of power posteriors. In the statistics literature the use of 
thermodynamic integration is detailed thoroughly by Neal~\citeyear{Neal1993} together with other techniques from statistical physics and furthermore by Gelman and 
Meng~\citeyear{GelmanMeng1998} and more recently by Friel and Pettitt~\citeyear{FrielPettitt2008}.

As in \cite{FrielPettitt2008}, for data $y$, parameters $\theta$ and temperature parameter $t\in [0,1]$, define the power posterior as the annealed distribution
\begin{equation}
	p(\theta|y,t) \propto f(y|\theta)^{t}p(\theta),
	\label{eq:powpost}
\end{equation}
which has normalising constant defined as
\begin{equation}
	z_t(y) = \int_{\theta}{f(y|\theta)^t p(\theta) d\theta}.
	\label{eq:zt}
\end{equation}
Throughout we assume that $p(\theta|y,t)$ is proper, so that $z_t(y)$ exists for all $t\in[0,1]$, in particular, this assumes a proper prior. 
Clearly, the evidence is realised when $t=1$, that is, $z_1(y) = p(y)$ and when $t=0$ the integration is over the prior with respect to $\theta$, thus $z_0(y) = 1$. In what follows 
we make use of the power posterior identity
\begin{equation}
	\log p(y) = \log \left\{ \frac{z_1(y)}{z_0(y)} \right\} = \int_{0}^{1}{\E_{\theta|y,t} \log{f(y|\theta)} }\;dt.
	\label{eq:powpost_identity}
\end{equation}
In fact, more generally one can express the log-evidence as
\begin{equation}
  \log p(y) = \E_{\theta,t|y} [\log f(y|\theta) / p(t)].
  \label{eqn:pp_iden}
\end{equation}
for some temperature prior $p(t)$. 
In practice the log-evidence is estimated, using a discretised temperature schedule, $t\in{[0,1]}$, $0=t_0<t_1,\dots,t_m=1$ and MCMC draws 
$\theta_j^{(i)}$ for $i = 1,2,\ldots,N$ from each power posterior $p(\theta|y,t_j)$, as
\begin{equation}
	\log p(y) \approx \sum_{j=1}^m \frac{(t_{j}-t_{j-1})}{2} \left( \E_{\theta|y,t_j} \log{f(y|\theta)} + \E_{\theta|y,t_{j-1}} \log{f(y|\theta)} \right).
	\label{eq:powpost_quadrature}
\end{equation}
Using a burn-in of $K<N$ iterations, $\E_{\theta|y,t_j} \log{f(y|\theta)}$ is estimated for a fixed $t_j$ by
\begin{equation}
	\E_{\theta|y,t_j} \log{f(y|\theta)} \approx \frac{1}{N-K}\sum_{j = K+1}^{N}{\log p(y|\theta_j^{(i)})}.
	\label{eq:estE}
\end{equation}

Alternatively, the updated power posterior estimate of Friel \textit{et al.}~\citeyear{FrielHurnWyse2013} employ a correction to the trapezoidal rule such that
\begin{eqnarray}
  \log p(y) \approx \sum_{j=1}^m \frac{(t_{j}-t_{j-1})}{2} \left( \E_{\theta|y,t_j} \log{f(y|\theta)} + \E_{\theta|y,t_{j-1}} \log{f(y|\theta)} \right) \nonumber \\ 
- \sum_{j = 1}^{m}\frac{(t_{j}-t_{j-1})^2}{12}\left(\V_{\theta|y,t_j} \log f(y|\theta) - \V_{\theta|y,t_{j-1}} \log f(y|\theta)\right),
\label{eq:powpost_quadrature_update}
\end{eqnarray}
where $\V_{\theta|y,t} \log f(y|\theta)$ is the variance of $\log f(y|\theta)$ with respect to the power posterior, $p(\theta|y,t)$. This approximation consistently out-performs 
the standard estimate with no additional computation cost. Indeed recent work by Oates \textit{et al.}~\citeyear{oates16} has shown that is possible to achieve further 
improvement through the use of control variates to efficiently estimate $\E_{\theta|y,t_j} \log{f(y|\theta)}$ and $\V_{\theta|y,t_j} \log{f(y|\theta)}$ for each temperature 
$t_j \in[0,1]$, at very little extra computational cost. Together with the numerical integration scheme (\ref{eq:powpost_quadrature_update}), the authors have shown 
that this can yield a dramatic improvement in the statistical efficiency of the estimate of the evidence. Finally we note that Hug \textit{et al.}~\citeyear{hug16} presents 
a further refinement of the power posterior approach which has shown to provide an improvement in the numerical integration over the temperature parameter.

\section{Widely applicable Bayesian information criterion}
\label{WBIC}

The widely applicable Bayesian information criterion (WBIC) \cite{Watanabe2013} promises to reduce the considerable computational burden that the method of power posteriors and 
indeed other evidence estimation methods suffer from. The key to WBIC is that there exists a unique temperature, say $t^* \in [0,1]$, such that, 
\begin{equation}
	\log p(y) = \E_{\theta|y,t^*} \log{f(y|\theta)}.
	\label{eq:WBIC}
\end{equation}
Hence, given this temperature $t^*$, only one Markov chain needs to be simulated at only one temperature value to estimate the evidence, using samples $\theta^{(i)}$ for 
$i = 1,2,\ldots,N$ from the power posterior $p(\theta|y,t^*)$ and equation (\ref{eq:estE}). 
The fact that equation~(\ref{eq:WBIC}) holds follows straightforwardly from the mean value theorem, which shows that there exists a particular $t^*$ such that
\begin{equation}
 \log p(y) = \frac{\log z_1(y) - \log z_0(y)}{1-0} = \frac{d}{dt} \log z_t(y) \bigg|_{t^*} = \E_{\theta|y,t^*} \log{f(y|\theta)}.
	\label{eq:WBIC2}
\end{equation}
Uniqueness of $t^*$ follows from the fact that $d^2 \log z_t(y) / dt^2 = \V_{\theta|y,t} \log{f(y|\theta)}$ and hence is strictly positive under standard regularity assumptions. 

In fact it is possible to provide an information theoretic interpretation of the optimal temperature $t^*$ in (\ref{eq:WBIC2}). Following
Vitoratou and Ntzoufras~\citeyear{vitoratou13}, it is straightforward to prove that $p_{t^*}$, the power posterior at the optimal temperature $t^*$, is equi-distant, in terms of 
Kullback-Leibler distance, from the prior and posterior. We can show this as follows.

Here, for brevity, we introduce the notation $p_t(\theta) = p(\theta|y,t)$ for all $t\in[0,1]$. 

\begin{lemma}
\label{lem:kl}
 The power posterior at the optimal temperature $t^*$ satisfies the identity
 \[
  KL(p_{t^*} || p_1) = KL(p_{t^*} || p_0).
 \]
\end{lemma}

\begin{proof}
Using the definition of Kullback-Leibler distance, we can re-write the statement of this lemma as, 
 \begin{eqnarray*}
  \int_{\theta} p_{t^*}(\theta) \log \frac{p_{t^*}(\theta)}{p_1(\theta)}\; d\theta &=& \int_{\theta} p_{t^*}(\theta) \log \frac{p_{t^*}(\theta)}{p_0(\theta)}\; d\theta  \\
  \iff \int_{\theta} p_{t^*}(\theta) \log p_1(\theta)\; d\theta &=& \int_{\theta} p_{t^*}(\theta) \log p_0(\theta)\; d\theta  \\
  \iff \int_{\theta} p_{t^*}(\theta) \log \frac{p_1(\theta)}{p_0(\theta)}\; d\theta &=& 0 \\
  \iff \int_{\theta} p_{t^*}(\theta) \log \frac{f(y|\theta)}{p(y)}\; d\theta &=& 0 \\
  \iff \log p(y) &=& \int_{\theta} p_{t^*}(\theta) \log f(y|\theta)\; d\theta \\
  &=& \E_{\theta|y,t^*} \log{f(y|\theta)},
 \end{eqnarray*}
which holds from equation~(\ref{eq:WBIC2}).
\end{proof}

This result may prove useful as a basis for estimating the optimal temperature. However, one should note that both the posterior, $p_1(\theta)$ and the power posterior 
at the optimal temperature, $p_{t^*}$ are generally intractable, leaving a direct evaluation of this identity unavailable.  Indeed we will use this result in Section~\ref{sec:theory:comparison} for an 
idealised analysis of a comparison of the mean squared error arising from the identity in (\ref{eq:WBIC2}) and that arising from (\ref{eqn:pp_iden}).  

Clearly, finding this optimal temperature $t^*$ is a challenging task. The main contribution of Watanabe~\citeyear{Watanabe2013} is to show 
asymptotically, as the sample size $n\rightarrow \infty$, that $t^*\sim 1/\log(n)$. WBIC is thus defined as
\begin{equation}
 WBIC = \E_{\theta|y,t_{w}}\log{f(y|\theta)} \approx \log p(y),
	\label{eq:WBIC3}
\end{equation}
where $t_{w} = 1/\log(n)$.

Watanabe~\citeyear{Watanabe2013} introduced WBIC in the context of algebraic geometry where it is applied to solve the problem of singularity in the statistical models commonly 
encountered for models with latent variables or hierarchical structure including, mixture models, hidden Markov models, neural networks and factor regression models. In the case 
of singular models (models where the mapping from parameters to a probability distribution is not one-to-one and where the Fisher information matrix is not always positive definite) the standard 
Bayesian information criterion (BIC or Schwarz criterion) \cite{Schwarz1978}  approximation to the marginal likelihood is known to be poor \cite{ChickeringHeckerman1997}. 
As before, the literature is lacking a comprehensive evaluation of the performance of WBIC in both regular and singular settings at finite $n$; this is our contribution below.

\section{A theoretical comparision of WBIC and power posteriors}
\label{sec:theory:comparison}

In this section we present a theoretical comparison of the mean squared error resulting from idealised implementations of both WBIC and the power posterior approach. 
Following (\ref{eqn:pp_iden}), thermodynamic integration or power posteriors (PP) is based on the identity 
$$
\log p(y) = \E_{\theta,t|y} [\log f(y|\theta) / p(t)].
$$
Here $p(t)$ is arbitrary, but for this comparison we suppose we have access to an an optimal choice (that minimises the Monte Carlo variance) and is given by
\begin{eqnarray*}
p^*(t) & = & \arg\min_{p(t)} \mathbb{V}_{\theta,t|y}[\log (p(y|\theta)) / p(t)] \\
& \propto & (\mathbb{E}_{\theta|t,y} [(\log p(y|\theta))^2])^{1/2}
\end{eqnarray*}
as shown in Calderhead and Girolami \citeyear{calder:giro09}.
It has been shown in numerous studies that this optimal choice can be well approximated using power law heuristics.

As before, the recently proposed WBIC, is based on the identity 
$$
\log p(y) = \mathbb{E}_{\theta|t^*,y}[\log p(y|\theta)]
$$ 
where $t^*$ is the solution to $\text{KL}(p_{t^*} || p_0) = \text{KL}(p_{t^*} || p_1)$ where $p_t$ denotes the power posterior $p(\theta|t,y)$. For this comparison we 
suppose we have access to $t^*$. 

We note here that WBIC is not a special (or degenerate) case of PP, but there are some visual similarities.
WBIC actually uses more information on the smoothness of the integrand, compared to PP, so from this simple perspective WBIC can be expected to perform better in principle.

Being expressed as expectations, both identities can, in principle, be used to produce unbiased estimates of $\log p(y)$ via Monte Carlo.
The natural theoretical question to ask is which Monte Carlo estimator has the lower variance.
An instructive analytical analysis of this idealised case is provided below.

\begin{case}
Consider the following simple example, taken from Friel and Pettitt \citeyear{FrielPettitt2008} and considered elsewhere by \shortcite{GelmanHwangVehtari2013}. Suppose data $y=\{y_i:\: i=1,\dots,n\}$ are 
independent and $y_i \sim N(\theta,1)$. Assuming an informative prior $\theta \sim N(m,v)$, this leads to a power posterior, $\theta|y,t \sim N(m_t, v_t)$ where
\begin{displaymath}
  m_t = \frac{nt\bar{y}+ m/v}{nt+ 1/v} \:\:\:\mbox{and}\:\:\:
  v_t = \frac{1}{nt+1/v}.
\end{displaymath}
It is straightforward to show that
\begin{equation}
  \E_{\theta|y,t} \log{f(y|\theta)}  = -\frac{n}{2}\log{2\pi} - \frac{1}{2}\sum_{i=1}^n (y_i-\bar{y})^2  
					- \frac{n}{2}\frac{(m-\bar{y})^2}{(vnt+1)^2} - \frac{n}{2}\frac{1}{(nt+1/v)}.
  \label{eq:toy}
\end{equation}
Moreover, it is easy to show that
\begin{equation}
 \log p(y) =-\frac{n}{2} \log{2\pi} - \frac{1}{2} \log{\frac{v}{v^*}} 
					- \frac{1}{2}\left[ \sum_{i=1}^{n} y_i^2 + \frac{m^2}{v} - \frac{(n\bar{y}+m/v)^2}{n+1/v} \right],
 \label{eqn:log_evid}
\end{equation}
where $v^* = \frac{1}{n + 1/v}$ is the posterior variance of $\theta$. 
This example is useful since not only does it allow analytical evaluation of (\ref{eq:toy}) and (\ref{eqn:log_evid}), but it also possible to analytically find the
optimal temperature $t^*$ as the solution to the identity in Lemma~\ref{lem:kl}. 

\end{case}


\begin{figure}[t]
\centering
\includegraphics[width = 0.32\textwidth,clip,trim = 4cm 9cm 4cm 9cm]{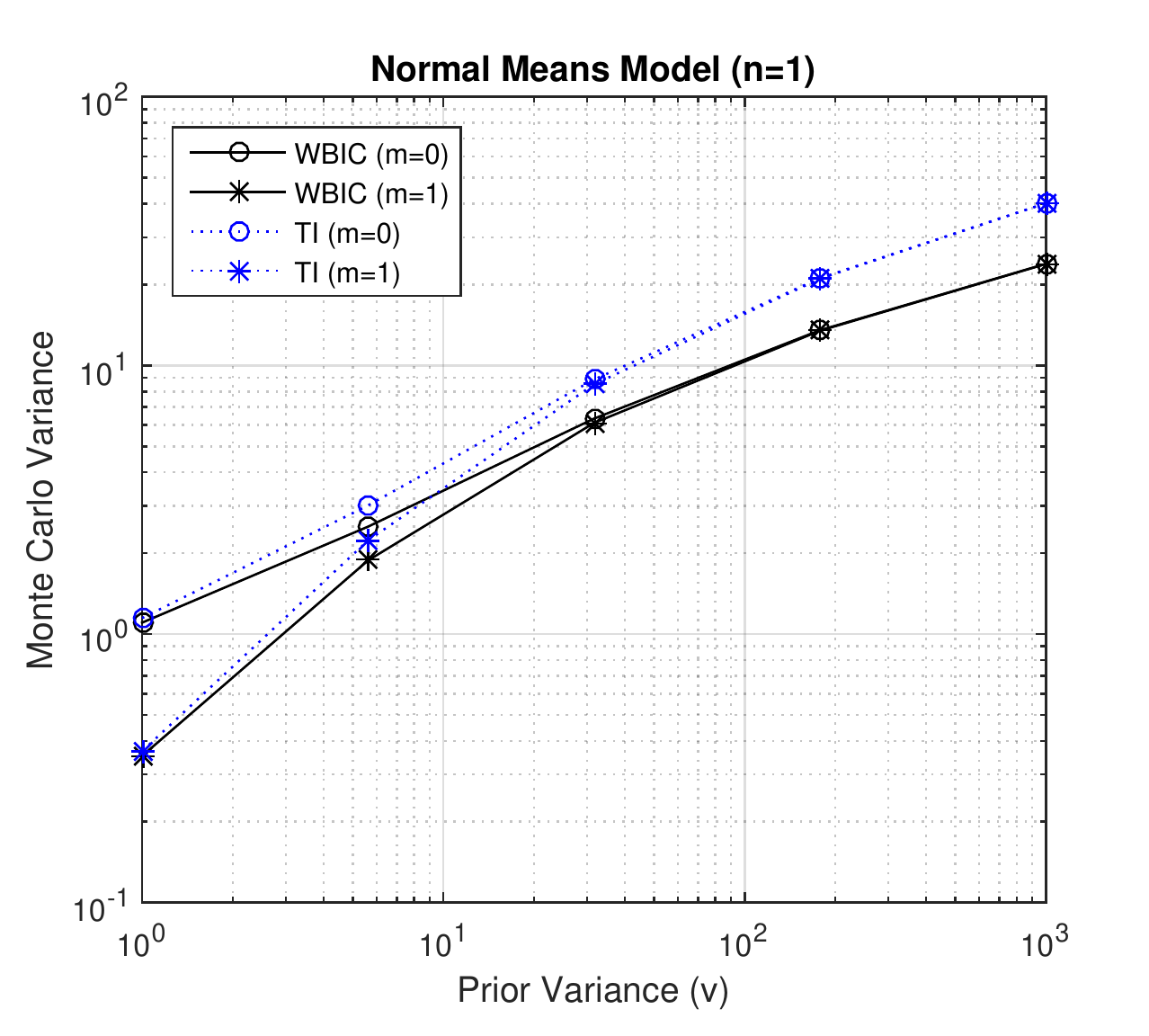}
\includegraphics[width = 0.32\textwidth,clip,trim = 4cm 9cm 4cm 9cm]{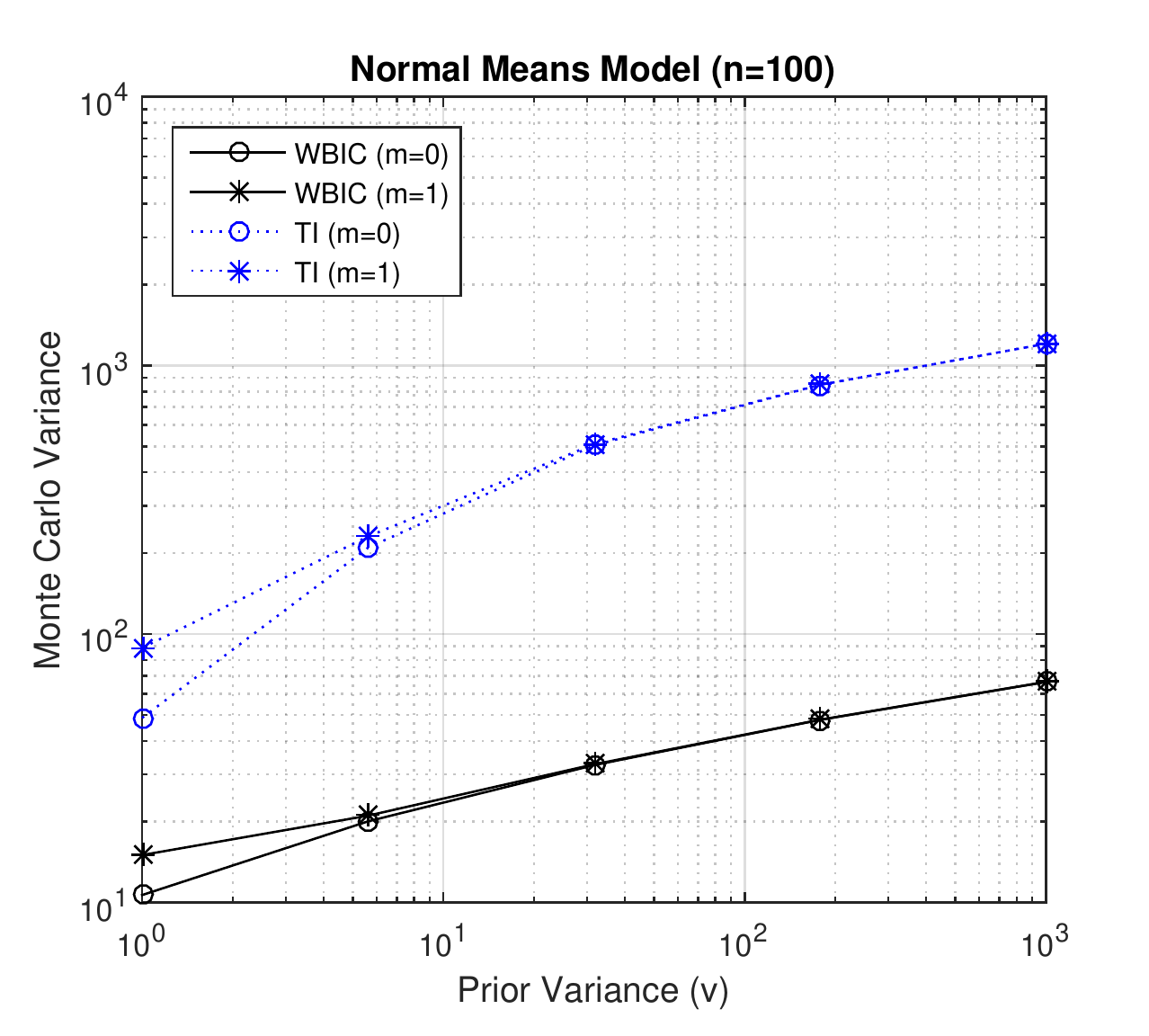}
\includegraphics[width = 0.32\textwidth,clip,trim = 4cm 9cm 4cm 9cm]{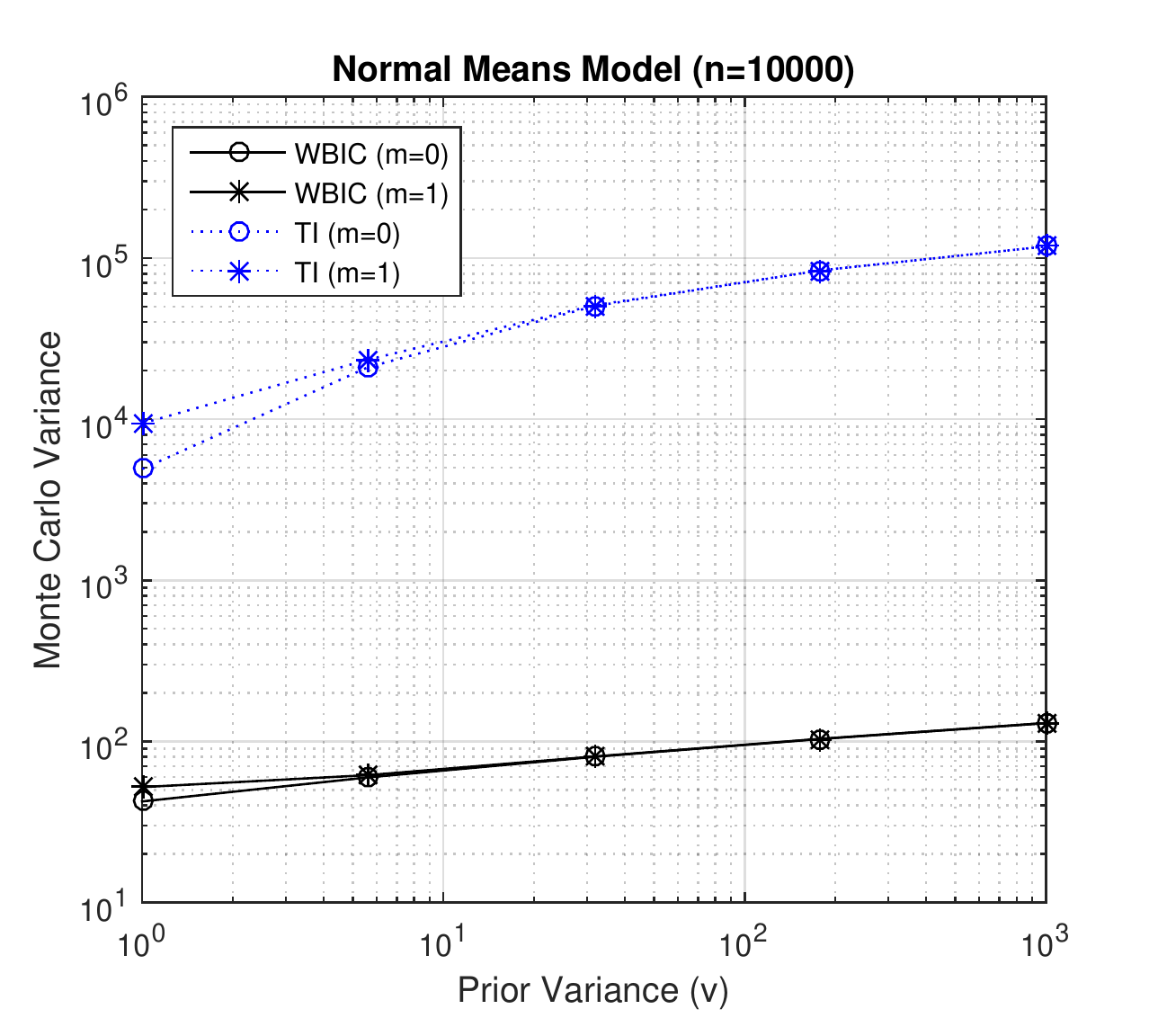}
\caption{Normal means model. Comparison of (idealised) thermodynamic and WBIC estimator variances.
Left: $n = 1$. Middle: $n = 100$. Right: $n = 10000$.
In each panel we show the case of prior mean $m \in \{0,1\}$ for varying values of the prior variance $v$.
Data $y$ were generated from the model with $\theta = 0$.}
\label{fig:example}
\end{figure}

Figure \ref{fig:example} displays the Monte Carlo variances of the (idealised) thermodynamic and WBIC estimators.
Data $y$ were generated from the model with $\theta = 0$ where a priori, $\theta\sim N(m,v)$ and we present results for $m \in \{0,1\}$ and for various values of $v$. 
For this example,
\[
 \mathbb{V}_{\theta|t^*,y}[ \log p(y|\theta)] < \mathbb{V}_{\theta,t|y}[\log (p(y|\theta)) / p^*(t)],
\]
for all combinations of $(y,m,v)$. 

For uninformative data ($n$ small; left panel) the PP and WBIC estimators have essentially equal Monte Carlo variance.
For informative data ($n$ large; right panel) the PP estimator variance can be orders of magnitude larger than the WBIC variance at all values of the prior variance $v$.
In all cases the Monte Carlo variance of PP increases relative to WBIC as the prior variance $v$ goes to infinity.
This reflects the fact that PP has to evaluate an expectation at $t = 0$, while the WBIC estimator deals with $t^* > 0$ (although $t^* \rightarrow 0$ as $n \rightarrow \infty$).

For this case study an optimally-tuned WBIC method has lower theoretical mean-squared error than an optimally-tuned PP method (again, ignoring practical details for the moment). 
The normal means model, whilst only one model, can be considered as a caricature of ``regular'' inference problems.
This example thus suggests that in scenarios with either informative data or very vague priors, for regular models, PP can severely under-perform an optimally configured WBIC estimator.

However, there are several factors that are relevant in practice that are not captured by the idealised formulation above:
\begin{enumerate}
\item The optimal temperature distribution $p^*(t)$ for thermodynamic integration is generally unavailable in closed-form.
\item It is common to use quadrature to integrate out $t$ in the thermodynamic integral, which can be more efficient than joint Monte Carlo sampling of $(\theta,t)$ but induces a bias into the estimate.
\item The $t^*$ in WBIC is unknown and a guess of $t_w = 1 / \log(n)$ is used in practice which may induce a significant error.
\item Several extensions of the power posterior approach have been developed which have improved the statistical efficiency of the estimator, \shortcite{FrielHurnWyse2013} and \shortcite{oates16}. 
Similar development and extension may also be possible for WBIC. 
\end{enumerate}

The above considerations motivate an empirical investigation of the relative merits of the two approaches which we now explore. 

\section{Empirical examples}
\label{Numerical}

We consider four examples where in all cases the motivation is to assess the performance of WBIC as an evidence approximation. The first model is one for which it is possible to calculate 
both $\log p(y)$ and WBIC analytically. The second model allows exact calculation of $\log p(y)$ only. The third model, logistic regression, is one where neither the log evidence nor 
WBIC can be evaluated exactly. The final model is a finite Gaussian mixture model, a singular model that WBIC was designed to handle and where neither the evidence nor WBIC can be evaluated exactly.
%
In all four models, the approximation $t_w = 1/\log(n)$ is used.

\subsection{A tractable normal model}
\label{SimpleModel}

\begin{case}
Consider again the tractable normal model from Section~\ref{sec:theory:comparison}.
Here $100$ datasets were simulated for each of the following values of $n=50,100,1000,10000$. Within each dataset $y_i\sim N(0,1)$ and a priori, $\theta\sim N(0,10)$. 
For each value of $n$, the optimal temperature $t^*$ was found by solving the identity in Lemma~\ref{lem:kl}. In addition, the temperature $t_w=\frac{1}{\log{n}}$ corresponding 
to WBIC was also recorded. 
\end{case}

The results are displayed in 
Figure~\ref{fig:tract_norm}(a). Clearly, as $n$ increases, as expected, the optimal temperature becomes closer to the temperature corresponding to WBIC. However, for relatively 
small values of $n$, there is typically a large discrepancy between $t_w$ and $t^*$. Through closer inspection of the curve of expected log deviances and the 
temperature, see Figure~\ref{fig:pine1} for example, it is clear that overestimates of the optimal temperature will not necessarily translate to large overestimates of the log 
evidence as the curve is reasonably flat as the temperature approaches one. 
\begin{figure*}[tbp] 
	\begin{center}
\begin{tabular}{cc}
		\includegraphics[angle = 270, scale=0.45]{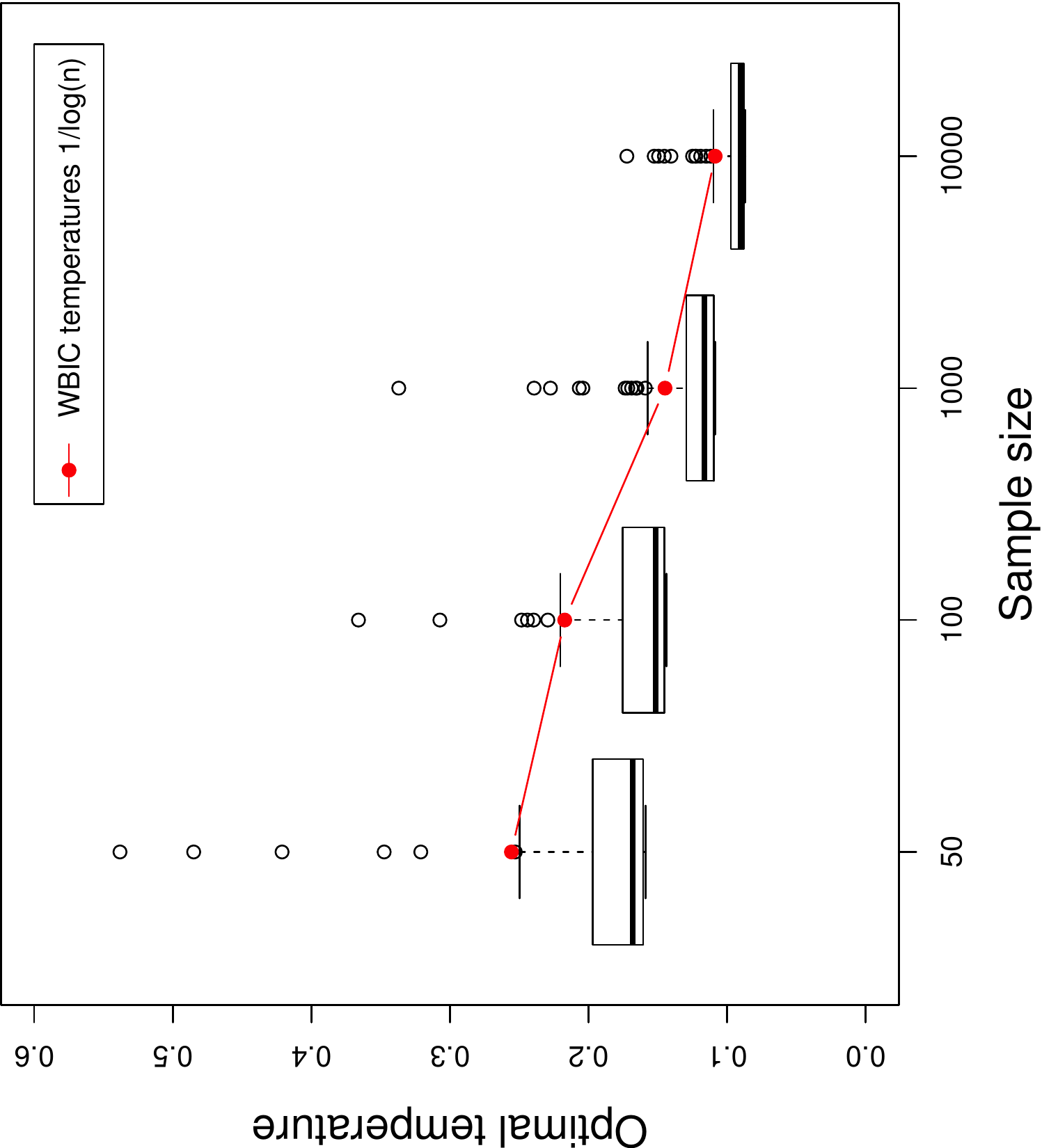}
		\includegraphics[angle = 270, scale=0.45]{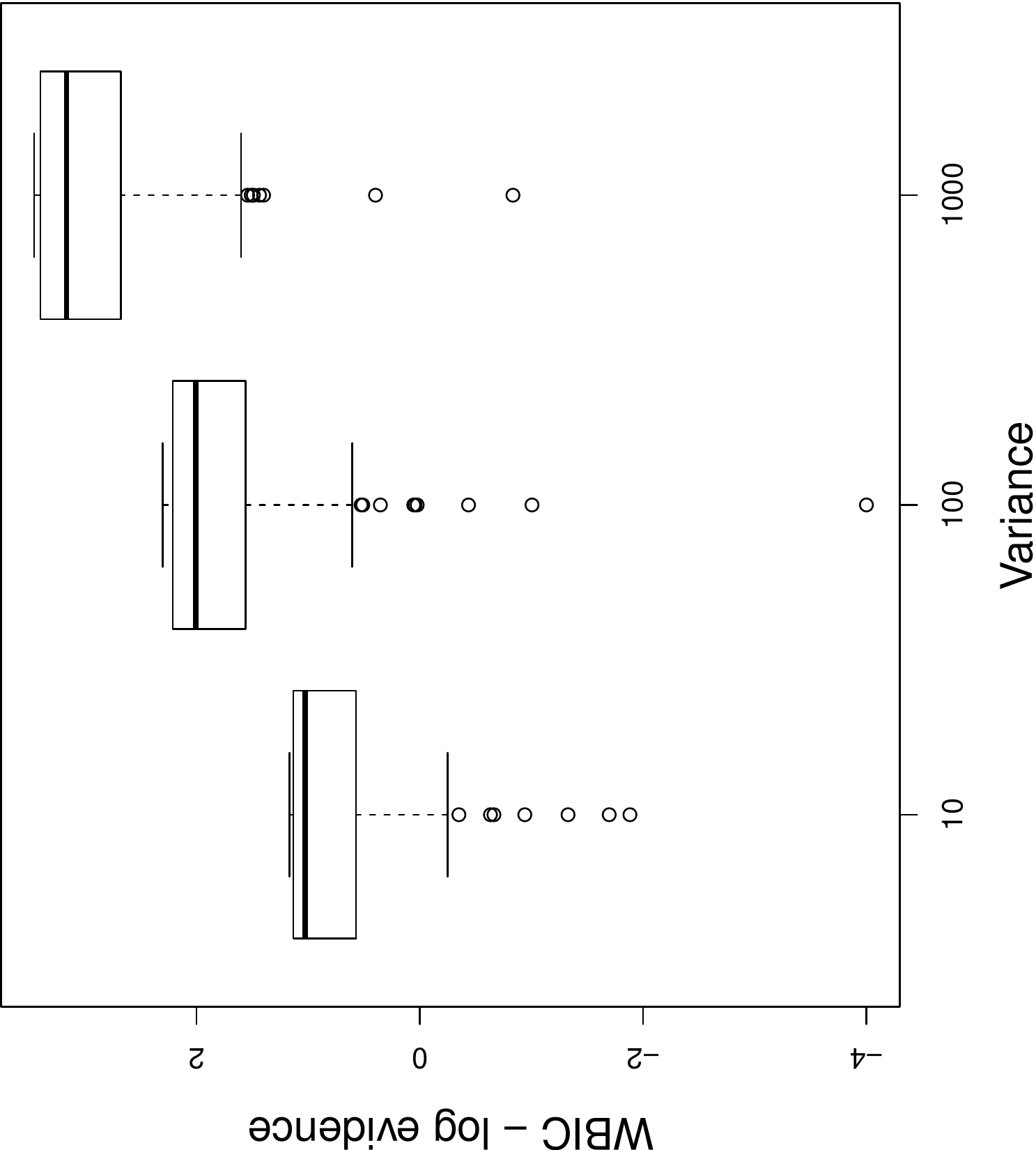}
\end{tabular}
		\caption{Tractable normal model: \textbf{(a)} Box-plots of the sample distribution of true temperature for $100$ datasets for various sizes $n$. Also displayed 
on each box-plot are the temperature corresponding to the WBIC estimate of the log evidence. \textbf{(b)} Box-plots of the difference between WBIC and the log evidence for $100$ 
independent datasets with $n=50$ observations and each	with different prior variances. It is clear that the difference grows larger as the variance grows and that the WBIC typically 
overestimates the log evidence}
		\label{fig:tract_norm}
	\end{center}
\end{figure*}
In Figure~\ref{fig:tract_norm}(b) for each value of $n$, the prior variance $v$ is now set as $10,100,1000$ with the data simulated as $y \sim N(0,1)$ as in (a) and it can be seen 
that, as expected, larger prior variances result in poorer WBIC estimates of the log evidence. Interestingly, WBIC tends to overestimate the log evidence.

Consider now the case of unit information priors \cite{KassWasserman1995} and considered subsequently in the context of the Bayesian information criterion (BIC or Schwarz criterion) 
by Raftery \citeyear{Raftery1999} and Volinsky and Raftery \citeyear{VolinskyRaftery2000} in sociology and survival models, respectively. A unit information prior represents the 
amount of information contained in one observation of the data, such priors can be quite informative and are used here to illustrate the applicability of WBIC to this model. 

In the present model, correct specification of the mean of a unit information prior, here $\theta \sim N(m,1)$, was vital to the performance of WBIC. Figure~\ref{fig:Comparem0m1} 
illustrates WBIC plotted against $\log p(y)$ for data simulated from $N(0,1)$ of size $n = 10000$ and a prior mean of $m=0$ or $m=1$ with unit variance in either case. The WBIC 
approximation to $\log p(y)$ is particularly bad for the case where $m=1$ (and this effect increases with sample size).  
\begin{figure}%
	\begin{center}
		\includegraphics[angle = 270,scale=0.65]{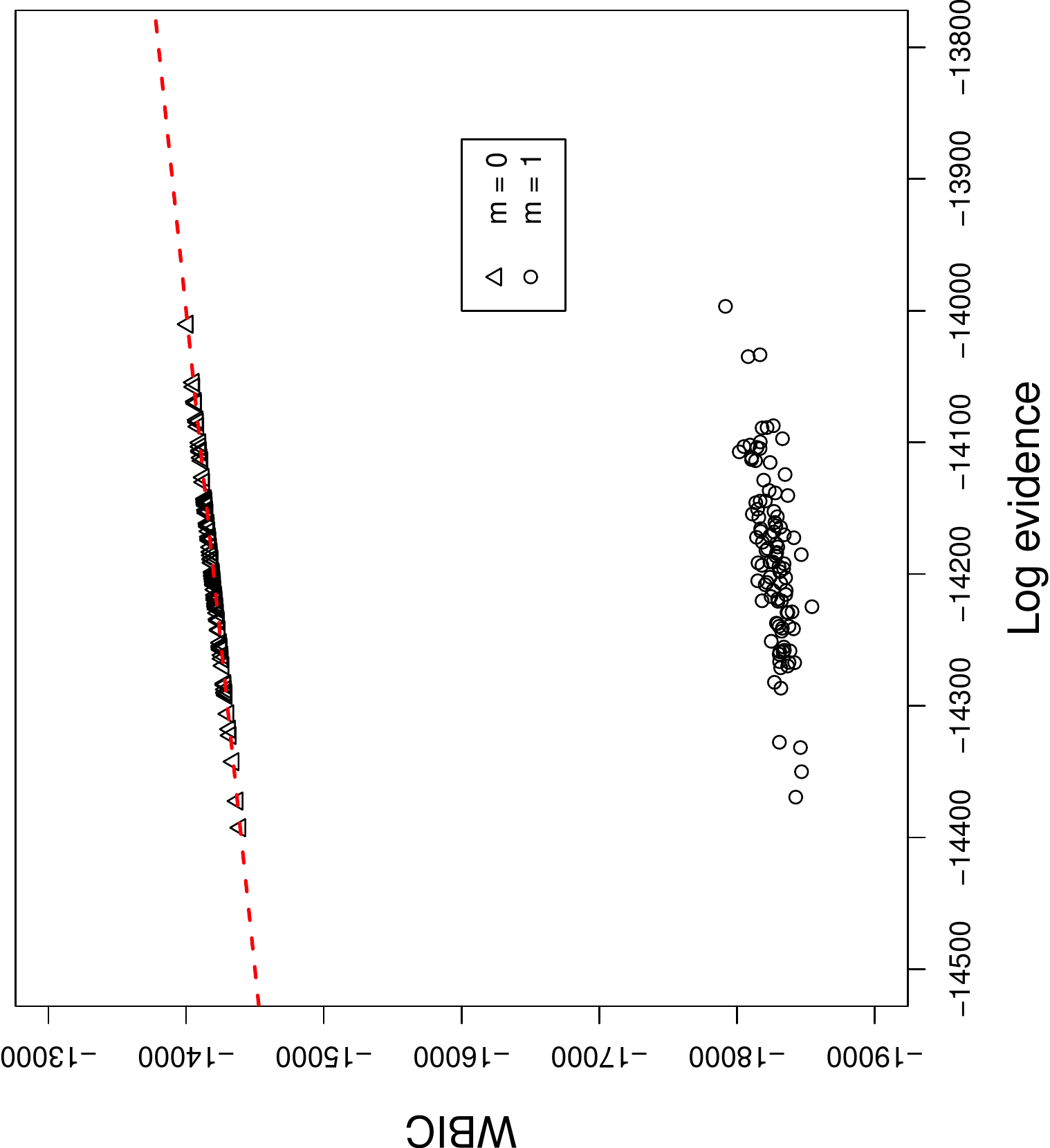}
	\end{center}
	\caption{Tractable normal model: WBIC against $\log p(y)$ for a mis-specified prior mean. One hundred datasets of size $n=1000$ are simulated from $y \sim N(0,1)$ and WBIC 
is evaluated at $t = \frac{1}{\log n}$ using the prior $\theta \sim N(m,1)$ where $m$ is $0$ or $1$. Note the stark difference in the log evidence and WBIC for the mis-specified model. 
The line where WBIC $= \log p(y)$ is marked on the figure}%
	\label{fig:Comparem0m1}%
\end{figure}

Therefore the question arises as to the appropriate prior mean for a unit information prior in this circumstance. The data informed prior $\theta \sim N(\bar{y},1)$, or equivalently define $\widetilde{y} = y - \bar{y}$ and the prior $\theta \sim N(0,1)$, is one obvious candidate for this. Though by this correction information about the mean is now wholly dependent on the data. 

An interesting observation can be made for fixed $n$, mean corrected data and the unit information prior $\theta \sim N(0,1)$, the difference between $t_w$ and $t^*$ is constant for every simulated dataset. That is, the simulation produces a deterministic result. Similarly, the difference between WBIC and $\log p(y)$ is also deterministic for every simulated dataset.

In Figure~\ref{fig:tempcomparison}, the optimal temperature $t^{*}$ satisfying equation (\ref{eq:WBIC}) is plotted against the temperature, $1/\log(n)$, for 
datasets of size $n = 3,4,5,\ldots,50$$,60,70,\ldots,100000$; a WBIC estimate with $t_w=1/\log(n)$ is of course not suitable for $n = 1,2$ and $t \in [0,1]$. 
The biggest differences occur for small $n$. 

\begin{figure*}%
	\begin{center}
		\begin{tabular}{cc}
				\includegraphics[angle = 270,scale=0.45]{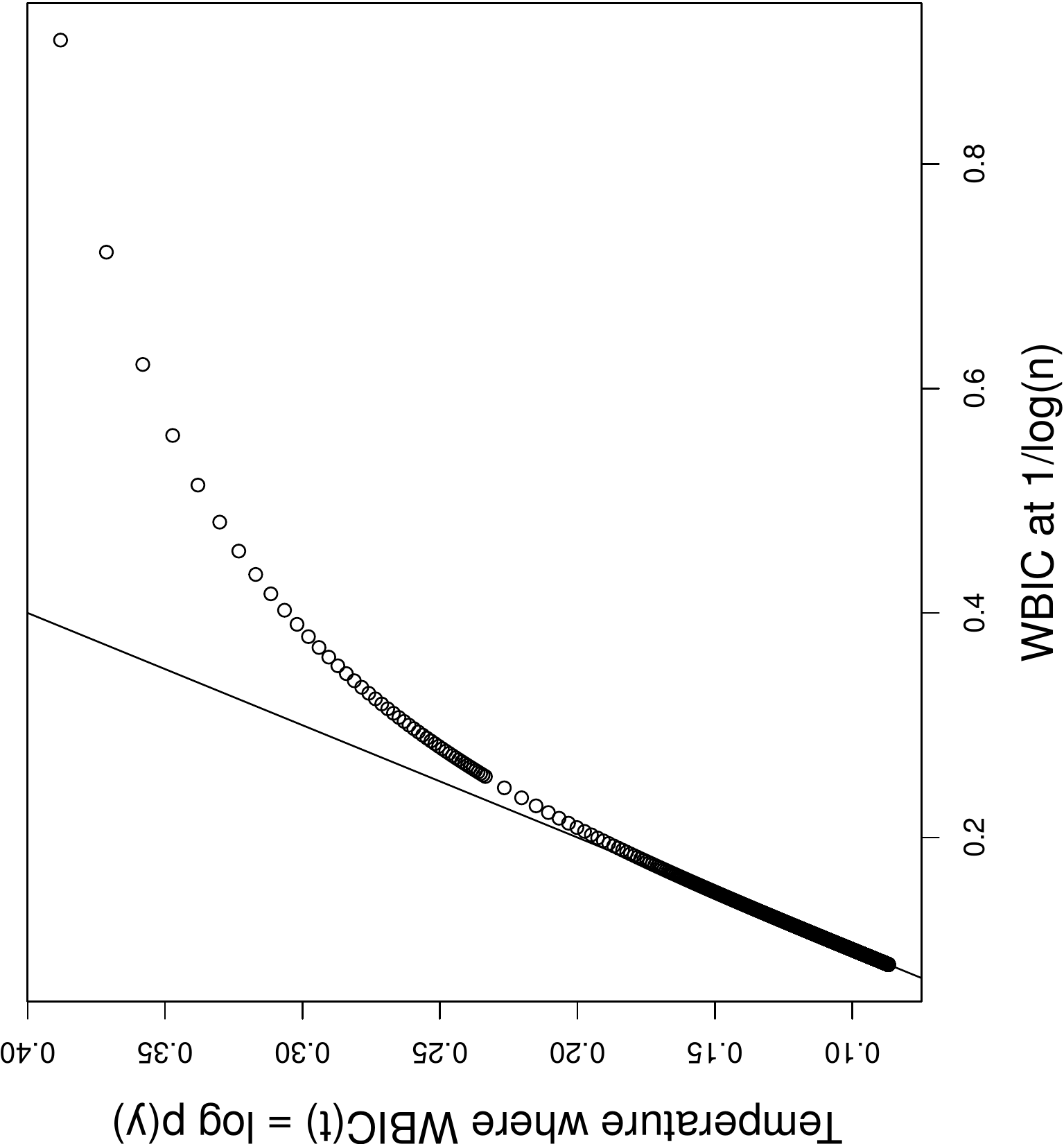} &
				\includegraphics[angle = 270,scale=0.45]{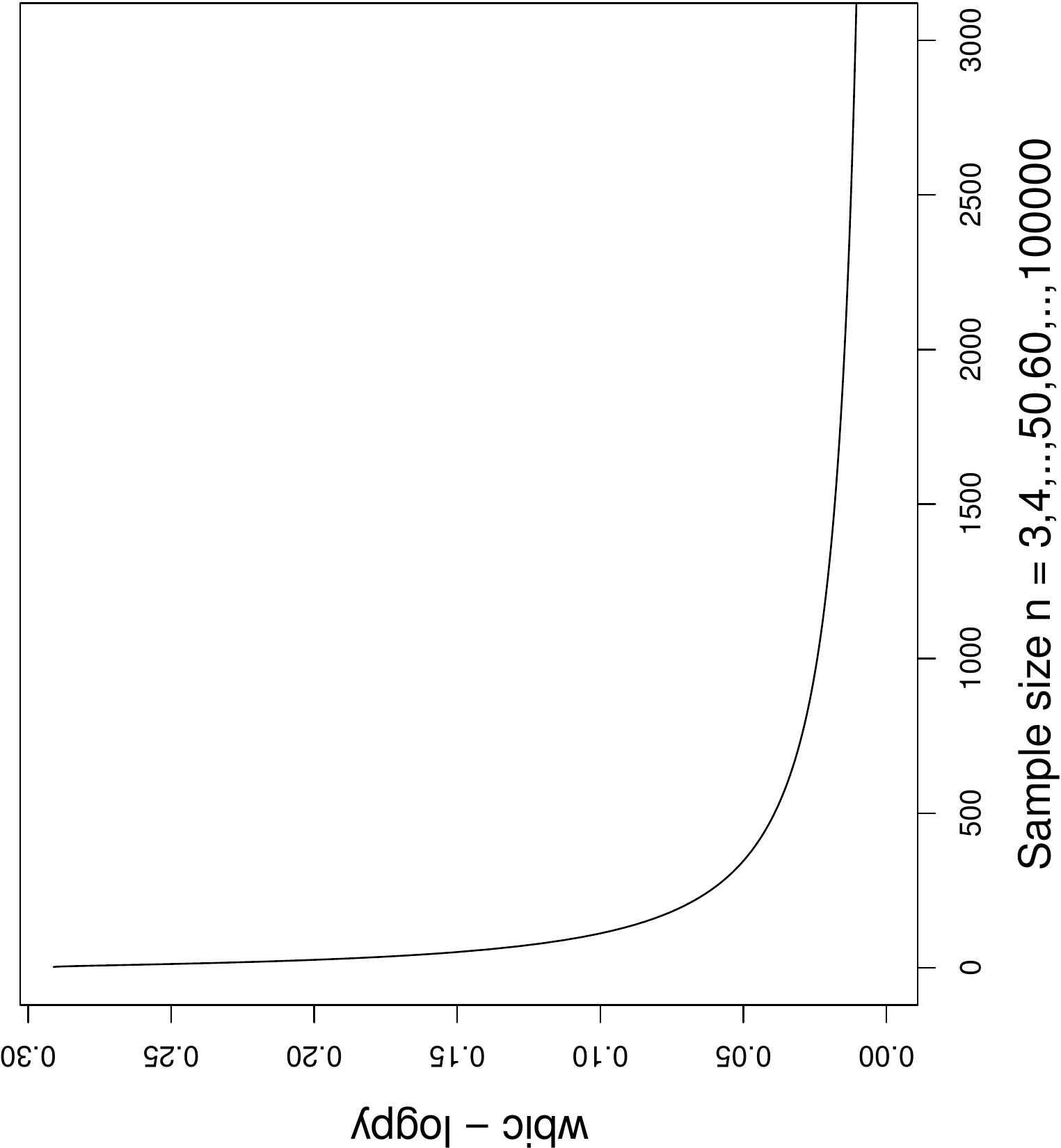} \\
		 (a) & (b) \\
		\end{tabular}
	\end{center}
	\caption{Tractable normal model: In both plots the data are mean corrected and a unit information prior is used, that is $\theta \sim N(0,1)$, and the sample 
	size varies from $n = 3$ to $n=100000$. \textbf{(a)} For each $n$, the temperatures are compared between the optimal temperature $t^*$ , (\ref{eq:WBIC}), plotted against $t = \frac{1}{\log(n)}$. 
Smaller values of $n$ incur a bigger difference in the two temperatures and the relationship is surprisingly regular as $t$ becomes large and the two temperatures become equal. 
The line $t = t^{*}$ is included in the figure. \textbf{(b)} The difference between WBIC and $\log p(y)$ is plotted against the sample size 
(up to $n = 3000$). Even for relatively small $n$, WBIC is accurate.}%
	\label{fig:tempcomparison}%
\end{figure*}


It is reassuring that the method performs admirably for the case of mean corrected data and a unit information prior. Though again, WBIC tends to slightly overestimate the log evidence in this case.

%
%
%

\subsection{Non-nested linear regression}
\label{LinearRegression}

Here we consider using WBIC to compute a Bayes factor and compare the results to existing methods to estimate the marginal likelihoods. 

\begin{case}
The data considered in this section describe the maximum compression strength parallel to the grain $y_i$, density $x_i$ and density adjusted for resin content $z_i$ for $n = 42$ 
specimens of {\it radiata} pine. Given the investigation of the tractable normal model, Sect.~\ref{SimpleModel}, WBIC is not expected to perform particularly well with such a small 
sample size though. These data originate from \cite{Williams1959}. It is wished to determine whether the density or the resin-adjusted density is a better predictor of compression 
strength parallel to the grain. With this in mind, two Gaussian linear regression models are considered;
\begin{equation}
	\begin{array}{lccc}
		\mbox{Model 1:} & y_i = \alpha + \beta(x_i - \bar{x}) + \epsilon_i, & \epsilon_i \sim N(0,\tau^{-1}), &\\
		\mbox{Model 2:} & y_i = \gamma + \delta(z_i - \bar{z}) + \eta_i, & \eta_i \sim N(0,\lambda^{-1}), & 
	\end{array}
\end{equation}
for $i = 1,\dots,n$. Under an informative set-up, the priors assumed for the line parameters $(\alpha,\beta)'$ and $(\gamma,\delta)'$ had mean $(3000,185)'$ with precision 
(inverse variance-covariance) $\tau Q_0$ and $\lambda Q_0$ respectively where $Q_0 = \mbox{diag}(r_0,s_0)$. The values of $r_0$ and $s_0$ were taken to be $0.06$ and $6$, respectively. 
A gamma prior with shape $a_0 = 6$ and rate $b_0 = 4\times300^2$ was taken for both $\tau$ and $\lambda$. These prior assumptions are broadly similar to the priors assumed for 
this data in other analyses. See for example \cite{FrielPettitt2008}.
\end{case}

It is possible to compute the exact marginal likelihood for both of these models due to the prior assumption that the precision on the mean of the regression line parameters is 
proportional to the error precision. For example, the marginal likelihood of 
Model 1 is given by
\begin{eqnarray}
	p(y) =  p^{-n/2} &b_0^{a_0/2}&\frac{\Gamma\left\{(n+a_0)/2 \right\}}{\Gamma\left\{ a_0/2\right\}} \nonumber \\
								&\times& \frac{|Q_0|^{1/2}}{|M|^{1/2}} \left(y' R y + b_0 \right)^{-(n+a_0)/2}
\end{eqnarray}
where $y  = (y_1,\dots,y_n)'$, $M = X'X + Q_0$ and $R  = I - X M^{-1} X'$ with the $i$th row of $X$ equal to $(1\,\,\, x_i)$ and $I$ is the $2\times 2$ identity matrix. 

The exact value of the Bayes factor of Model $2$ over Model $1$ is given in Table~\ref{tab:comparison} to show a comparison 
with other approaches to estimating the evidence and Bayes factor. 
This example was examined in detail in \cite{FrielWyse2012} and we refer the reader to this paper for precise details of how these methods were
implemented. The key 
point to take from this is that WBIC is reasonably competitive with the other methods, but at a significantly reduced computational overhead cost.

\begin{table}
	\begin{center}
		\begin{tabular}{l c r}
			\hline\hline
			Method \qquad & \qquad $\mbox{mean}(BF_{21})$ \qquad& \qquad$({\mbox{S.E.}}(BF_{21}))$\qquad\\
			\hline
			Exact 												& 4553.65 & $-$					\\
			Laplace approximation 				& 4553.63 & $-$					\\
			Laplace approximation MAP 		& 4553.74 & (1.05)	 		\\
			Harmonic mean estimator 			& 3718.57 & (909.17) 		\\
			Chib \& Jeliazkov's method 		& 4553.72 & (0.66)			\\
			Annealed importance sampling 	& 4542.43 & (140.27)		\\
			Power posteriors 							& 4535.11 & (74.75)			\\
			Nested sampling 							& 6817.52 & (6980.82)		\\
			WBIC 													& 4469.11 & (372.15) 		\\
			\hline
		\end{tabular}
	\end{center}
	\caption{Radiata Pine: Comparison of different approaches to estimating the Bayes factor of Model $2$ over Model $1$ based on $20$ runs of each algorithm for the Radiata Pine data.} 
	\label{tab:comparison}
\end{table}

Figure~\ref{fig:pine1} plots the expected log deviance with respect to $p(\theta|y,t)$ versus the temperature $t$. A fine grid of discrete temperatures in the range $[0,1]$ is 
employed and $\E_{\theta|y,t^*}[\log{f(y|\theta)}]$ is estimated for each $t_i\in [0,1]$ using a long MCMC run targeting the power posterior $p(\theta|y,t_i)$  The vertical 
line on the left hand side corresponds to the WBIC temperature $t_w =\frac{1}{\log(42)}=0.267$. The vertical line on the left hand side plots the temperature ($t^*\approx 0.19$) 
corresponding the true value of the log evidence. 
\begin{figure}[tbp] 
	\begin{center}
		\includegraphics[scale=0.65]{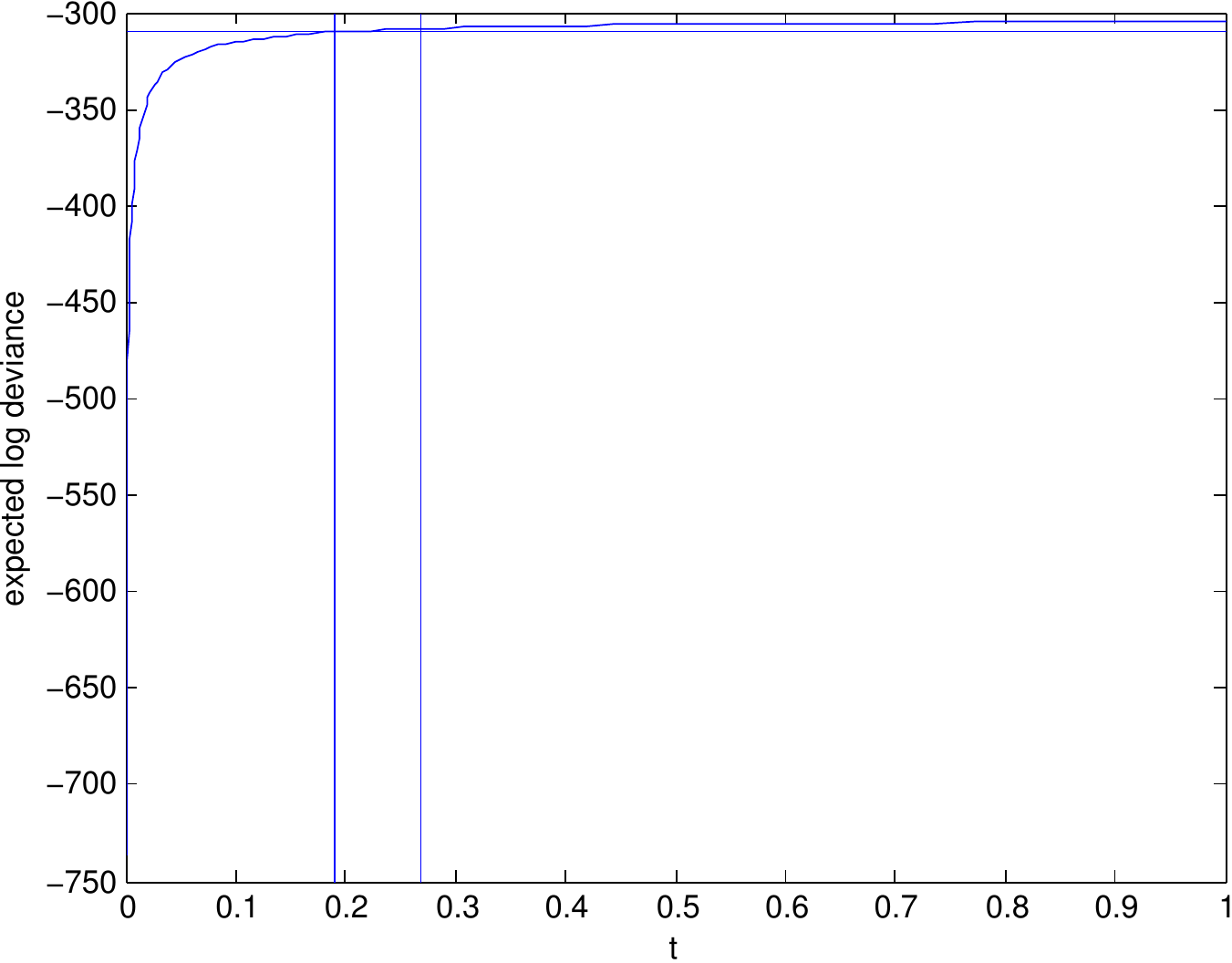}
		\caption{Pine data: Expected log deviance vs temperature. The vertical line on the left shows the temperature corresponding to the true evidence. The vertical line 
on the right corresponds to the temperature at the WBIC estimate of the evidence.}
		\label{fig:pine1}
	\end{center}
\end{figure}

Table~\ref{tab:pine_WBIC} shows the systematic bias in the estimates of the log evidence using WBIC, based on $20$ independent MCMC runs. 
\begin{table}
	\begin{center}
		\begin{tabular}{l l l}
			\hline
			& Model $1$ & Model $2$ \\
			True $\log p(y|m)$ & $-310.1283$ & $-301.7046$ \\
			mean WBIC & $-308.3390$ & $-299.8437$ \\
			s.e WBIC & $0.01$ & $0.01$ \\
			\hline
		\end{tabular}
	\end{center}
	\caption{Radiata pine data: WBIC estimate of the marginal likelihood for Model $2$ and Model $1$ compared to the true value of the log marginal likelihood for each method, based on $20$ independent runs.} 
	\label{tab:pine_WBIC}
\end{table}


We now consider the same evidence comparisons as those made in Table \ref{tab:pine_WBIC} but under a unit information prior formulation. In this regression model we re-parameterise the hyper-parameters as
\begin{equation*}
	(\alpha, \beta)' = (\gamma, \delta)' = \left[(X_1'X_1)^{-1}X_1'y + (X_2'X_2)^{-1}X_2'y\right]/2' 
\end{equation*}
which is very similar to the previous values of (3000,185)' considered above. The precision matrix $Q_0$ is now defined as
\begin{equation*}
	Q_0 = n\left[(X_1'X_1)^{-1} + (X_2'X_2)^{-1} \right]/2 \approx 
	\begin{pmatrix}
		3528 	& 0 \\
		0 		& 72945.73
	\end{pmatrix},
\end{equation*}
where $p_1 = p_2 = 2$ is the number of covariates in each model's design matrix. The variance parameters $\lambda$ and $\tau$ share hyper-parameters
\begin{equation}
	a_0 = 1, \hspace{20pt} b_0 = \frac{1}{n}y'\bar{R}y = 9701264,
\end{equation}
where $\bar{R} = (R_1 + R_2)/2$, $R_i = I_n - X_iM_i^{-1}X_i'$ and $M_i = X_i'X_i + Q_0$.

The evidence estimates under the unit information prior parameters are presented in Table \ref{tab:logEvidcomparison_unit_information1}. The values in the table are found from $20$ runs under each method. The harmonic mean and nested sampling estimates do not perform as poorly as in the previous set-up from Table \ref{tab:comparison} and the WBIC seems to be comparable with these methods, which are not significantly different from the true evidence values for both models. The standard error estimates for all models is markedly reduced from those found for the vague prior $BF_{21}$, though still unreliably high for the harmonic mean and nested sampling estimates. The $BF_{21}$ for WBIC is unsatisfactory however given the standard error estimate. Recall there are $n=42$ observations for these data and as seen with the tractable normal model, this seems too small for the WBIC estimate to perform well.

\begin{table}
	\centering
\resizebox{\columnwidth}{!}{
	\begin{tabular}{lrrrrrr}
		\hline\hline
		Method 							 					& $\mbox{mean}(\log p(y|M_1))$ &(S.E) & $\mbox{mean}(\log p(y|M_2))$ &(S.E)	& $\mbox{mean}(BF_{21})$ &(S.E.)  \\
		\hline                                                    
		Exact 												& -327.23 	& (--)         						& -324.93 &(--)                     	& 9.97 	& (--)             				\\
		Laplace approximation					& -327.23 	&(--)        							& -324.93 & (--)                			& 9.98 	& (--)   					 				\\
		Laplace approximation MAP			& -327.24 	&	(0.00012)    						& -324.93 & (0.000079)             		& 9.99  &(0.0015)									\\
		Harmonic mean estimator				& -331.09 	& (0.92)       						& -328.96 & (0.98)                 		& 21.24 &(43.67)									\\
		Chib \& Jeliazkov's method		& -327.23 	& (0.00014)    						& -324.93 & (0.000087)             		& 9.99  &(0.0018)									\\
		Annealed importance sampling  & -327.23 	& (0.038)      						& -324.95 & (0.034)                		& 9.77  &(0.40)										\\
		Power posteriors          		& -327.23 	& (0.017)      						& -324.93 & (0.019)                   & 9.99  &(0.26)										\\
		Nested sampling          			& -328.26 	& (1.49)       						& -326.29 & (1.48)                    & 31.55 &(62.25)									\\
		WBIC								         	& -331.04 	& (0.12)       						& -329.54	& (0.15)                    & 4.54  &(0.82)										\\
		\hline
	\end{tabular}
}
	\caption{Radiata Pine: Estimated log marginal likelihoods for each model and corresponding Bayes factors for each method under a unit information prior. These estimates are based on $20$ runs of the algorithm} 
	\label{tab:logEvidcomparison_unit_information1}
\end{table}

\subsection{Logistic regression models}
\label{logistic}

\begin{case}
Here we consider the Pima Indians dataset. These data contain instances of diabetes and a range of possible diabetes indicators for $n = 532$ Pima Indian women aged $21$ years or over. 
There are seven potential predictors of diabetes recorded for this group; number of pregnancies (\texttt{NP}); plasma glucose concentration (\texttt{PGC}); diastolic blood pressure 
(\texttt{BP}); triceps skin fold thickness (\texttt{TST}); body mass index (\texttt{BMI}); diabetes pedigree function (\texttt{DP}) and age (\texttt{AGE}). This gives $129$ potential 
models (including a model with only a constant term). Diabetes incidence ($y$) is modelled by the likelihood
\begin{equation}
	f(y|\theta) = \prod_{i=1}^n p_i^{y_i}(1-p_i)^{1-y_i}
\end{equation}
where the probability of incidence for person $i$, $p_i$, is related to the covariates (including constant term) $x_i = (1,x_{i1},\dots, x_{id})'$ and the parameters 
$\theta = (\theta_0,\theta_1,\dots,\theta_d)'$ by
\begin{equation}
	\log\left(\frac{p_i}{1-p_i}\right) = \theta' x_{i}
\end{equation}
where $d$ is the number of explanatory variables. An independent multivariate Gaussian prior is assumed for the elements of $\theta$, so that
\begin{equation}
	p(\theta) =  \left(\frac{\tau}{2p}\right)^{d/2} \exp\left\{-\frac{\tau}{2}\theta'\theta \right\}.
\end{equation}
The covariates were standardized before analysis. 
\end{case}

A long reversible jump run \cite{Green1995} revealed that the two models with the highest posterior probability were
\begin{align}
		\mbox{Model 1: } \mbox{\texttt{logit(p)}} &= 1 + \mbox{\texttt{NP}} + \mbox{\texttt{PGC}}	+ \mbox{\texttt{BMI}} + \mbox{\texttt{DP}} \nonumber \\
		\mbox{Model 2: } \mbox{\texttt{logit(p)}} &= 1 + \mbox{\texttt{NP}} + \mbox{\texttt{PGC}}	+ \mbox{\texttt{BMI}} + \mbox{\texttt{DP}} + \mbox{\texttt{AGE}}.	\label{eq:pima}
\end{align}
This reversible jump algorithm assumed a non-informative value of $\tau=0.01$ for the prior on the regression parameters. For this value of $\tau$ we carried out a reduced reversible 
jump run restricting to jumps only between these two models. The prior probabilities of the models were adjusted to allow for very frequent jumps (about 29\%). This gave a Bayes 
factor $BF_{12}$ of 13.96 which will be used as a benchmark to compare the other methods to. 

Table~\ref{tab:pima} displays results of estimates of the evidence for both models which were also implemented for this example in \cite{FrielWyse2012}. 
Here the WBIC estimate is not as competitive with the more computationally demanding methods. 
\begin{table}
\resizebox{\columnwidth}{!}{	
 \begin{tabular}{lccccccr}
 \hline \hline
   Method & $\mbox{mean}(\log p(y|M_1))$ &(S.E) & $\mbox{mean}(\log p(y|M_2))$ &(S.E) & $\mbox{mean}(BF_{12})$ &(S.E.) & Relative speed \\
   \hline
   Laplace approximation & -257.26 & (--) & -259.89 & (--) & 13.94 & (--) & 1\\
   Laplace approximation MAP & -257.26 & (0.015) & -259.93 & (0.026) & 14.44 & (0.50) & 108\\
   Harmonic mean estimator & -279.50 & (0.65) & -285.37 & (0.58) & 487.66 & (567.39) & 108\\
   Chib \& Jeliazkov's method & -257.23 & (0.02) & -259.86 & (0.02) & 13.89 & (0.46) & 44\\
   Annealed importance sampling & -257.29 & (0.54) & -260.38 & (0.68) & 34.04 & (35.34) & 194 \\
   Power posteriors & -260.91 & (0.14) & -264.06 & (0.12) & 23.82 & (4.66) & 184\\
   Nested sampling & -256.64 & (1.36) & -260.96 & (2.39) & 75.19 & (206.80) & 808\\
   WBIC & -251.49 & (0.63) & -253.49 & (0.45) & 9.30 & (6.46) & 17\\
   \hline
 \end{tabular}
}
\caption{Pima dataset: Estimated log marginal likelihoods for each model and corresponding Bayes factors for each method along with relative run times with $\tau = 0.01$. The standard error estimates are based on $20$ runs of the algorithm.} 
\label{tab:pima}
\end{table}

Figure~\ref{fig:pima1} displays a 'close-up' plot of temperature versus expected log deviance for a small range of temperatures. MCMC was used to estimate the expected log-deviance 
at each powered posterior. Again, as for the previous example, the temperature $t^*$ such that equation (\ref{eq:WBIC}) is satisfied is smaller than $1/\log{n}$. 
\begin{figure}[tbp] 
	\begin{center}
		\includegraphics[scale=0.65]{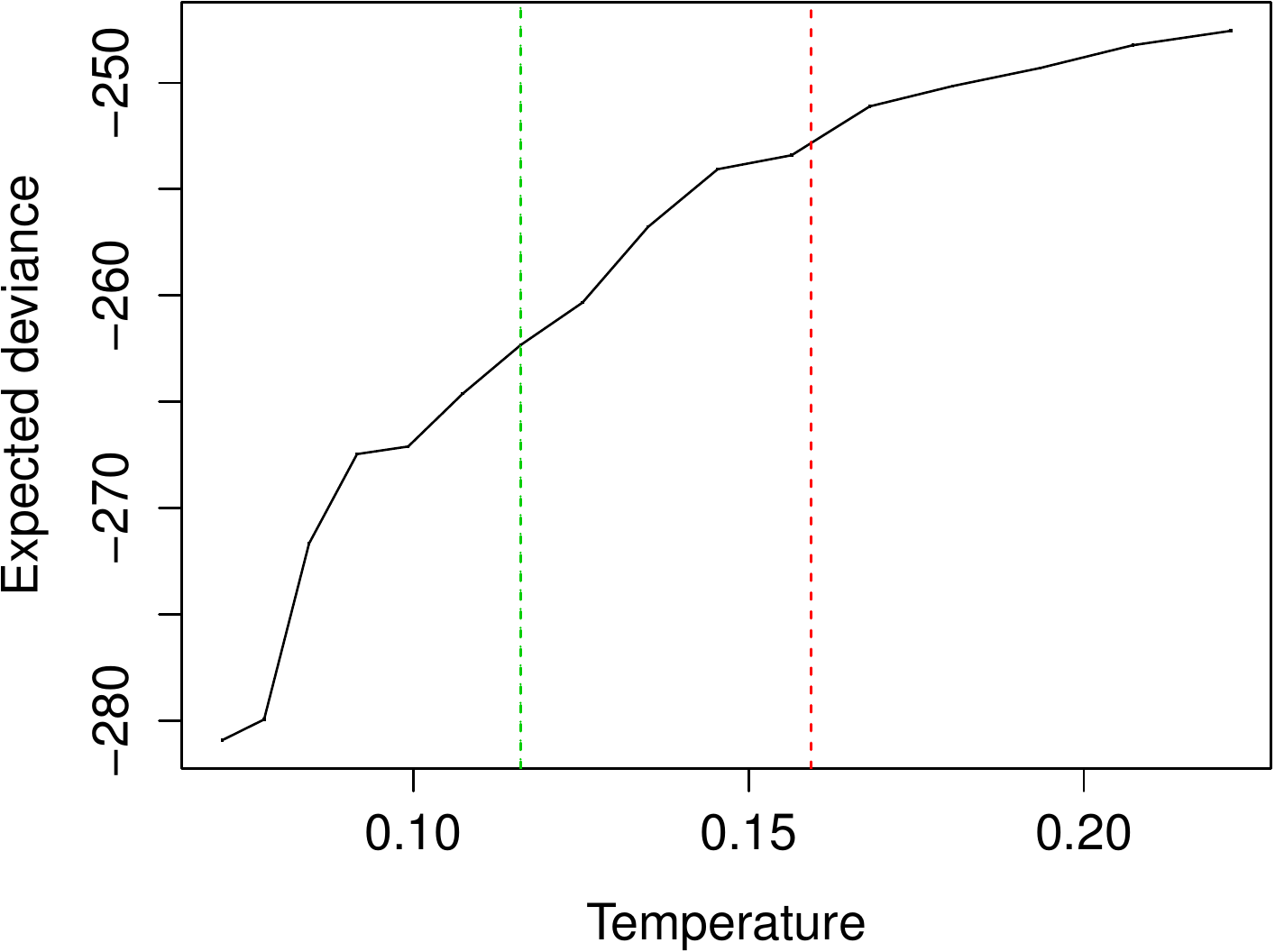}
		\caption{Pima dataset: Expected log deviance vs temperature. The green vertical line on the left shows the temperature $t^*$ corresponding to
		the true evidence. The red vertical line on the right corresponds to the temperature $t_w$ at the WBIC estimate of the evidence.}
		\label{fig:pima1}
	\end{center}
\end{figure}

As before, we now consider the performance of the evidence estimation techniques under a unit information prior formulation. As the models have different numbers of parameters the hyper-parameters are slightly different for the two models under comparison, which, of course, was also necessary in the estimates presented in Table \ref{tab:pima}. The prior mean for each model is defined as 
\begin{equation*}
	\theta_i \sim \mbox{MVN}\left(\mbox{MLE}(\theta_i), \frac{(X_i'X_i)^{-1}}{n} \right), \hspace{20pt} i = 1 \mbox{ or } 2.
\end{equation*}

The evidence estimates under the unit information prior are given in Table \ref{tab:logEvidcomparison_unit_information2}. The values in the table are based on $20$ runs 
of the algorithm. Ignoring the harmonic mean estimate, the results are quite striking. All the log evidence estimates are very similar and the standard errors are 
extremely small. 
One might expect that WBIC would perform well here given the sample size of $n=532$. 
However, Table~\ref{tab:pima} shows that WBIC is not as competitive as the other competing methods. However, when unit information priors are used, the results in
Table~\ref{tab:logEvidcomparison_unit_information2} show that WBIC performs as well as all of the other methods under consideration. 


\begin{table}
	\centering
\resizebox{\columnwidth}{!}{
	\begin{tabular}{l rrrrrr}
		\hline\hline
		Method 							 & $\mbox{mean}(\log p(y|M_1))$ &(S.E) & $\mbox{mean}(\log p(y|M_2))$ &(S.E)	& $\mbox{mean}(BF_{12})$ &(S.E.) \\
		\hline                                                    
		Laplace approximation					& -244.34 & (--)         		& -242.96 & (--)                  			& 0.25 & (--)    	 				\\
		Laplace approximation MAP			& -244.41 & (0.025)       	& -243.05 & (0.036)                			& 0.26 & (0.013)						\\
		Harmonic mean estimator				& -214.30 & (2.89)          & -203.66 & (2.29)                    	& 0.0021 & (0.0082)			\\
		Chib \& Jeliazkov's method		& -244.39 & (0.21)       		& -243.23 & (0.25)                			& 0.33 & (0.11)						\\
		Annealed importance sampling  & -244.34 & (0.0012)   	    & -242.96 & (0.0015)                   	& 0.25 & (0.00048)				\\
		Power posteriors         			& -244.34 & (0.0041)        & -242.96 & (0.0025)                    & 0.25 & (0.0011)					\\
		Nested sampling	          		& -244.34 & (0.0016)        & -242.96 & (0.0018)                   	& 0.25 & (0.00055)				\\
		WBIC								         	& -244.34 & (0.0011)        & -242.96	& (0.0010)                    & 0.25 & (0.00048)				\\
		\hline
	\end{tabular}
}
	\caption{Pima dataset: Log marginal likelihoods for each method and the Bayes factor estimates under a unit information prior. The mean and standard error values in the Table are based on $20$ runs of the algorithm} 
	\label{tab:logEvidcomparison_unit_information2}
\end{table}

\subsection{Finite Gaussian mixture model}
\label{mixture}


Watanabe introduced WBIC with the goal of approximating the evidence for singular statistical models. Here we present an analysis of WBIC for one such singular model, namely a finite
mixture model. 

\begin{case}
Consider now a finite mixture of $K$ components where for $i = 1,\ldots,n$ and with 
$n_k$ the number of observations in the $k$th component $(\sum_{k=1}^{K}{n_k} = n)$ there exist observations $\y = (y_1,\ldots,y_n)$. Conditioned on a set of labels 
$\z = (z_1,\ldots,z_n)$ satisfying $p(z_i = k) = w_k$ with $\sum_{k=1}^{K}{w_k} = 1$ the likelihood is given by


\begin{equation}
	p(\y|\mu, \sigma^2, \z) = \prod_{i = 1}^{n}{\sum_{k = 1}^{K}{ I_{\{z_i=k\}} \frac{1}{\sqrt{2\pi\sigma_k^2}}\exp{\left(-\frac{1}{2\sigma_k^2}(y_i - \mu_i)^2\right)} }}.
	\label{eq:likemix2}
\end{equation}
Where $\mu=\{\mu_1,\dots,\mu_K\}$, $\sigma^2 = \{\sigma_1^2,\dots,\sigma_K^2\}$ and $I_{\{z_i=k\}}$ denotes the indicator function taking the 
value $1$ when $z_i=k$ and zero, otherwise. Prior distributions are assigned to all model parameters as follows:
\begin{eqnarray*}
 \mu_k|\sigma_k^2 &\sim& N(\mu_0,\sigma_0^2), \\
 \sigma_k^2 &\sim& Ga(\alpha_0,\beta_0), \;\; k=1,\dots,K \\
 w &\sim& Dir(\alpha,\dots,\alpha) \\
 z_i &\sim& Multinomial (w),\;\; i=1,\dots,n.
\end{eqnarray*}

With observations in the $k$th component given by $C_k$, the full-conditional distributions for parameters $\mu_k, \sigma_k^2, z_i$ are given by
\begin{eqnarray}
	&&z_i = k|y_i, \mu_k, \sigma_k^2, w_k \propto w_k\frac{1}{\sqrt{2\pi\sigma_k^2}}\exp{\left(-\frac{1}{2\sigma_k^2}(y_i - \mu_j)^2\right)} \nonumber \\
	&&\w|\y,\z \sim \mbox{Dir}(\alpha + n_1, \ldots, \alpha + n_k) \nonumber \\
	&&\mu_k|\y, \sigma_0^2, \z \sim N\left(m_k, s_k^2 \right) \nonumber \\
	&&\sigma_k^2|\y, \mu_k, \z \sim \mbox{Inverse-Gamma}\left(\alpha_0 + n_k/2, \beta_0 + \sum_{i \in C_k}{(y_i - \mu_k)^2}/2\right)	
	\label{eq:fullconds2}
\end{eqnarray}
with $m_k = s_k^2\left( \frac{\mu_0}{\sigma_0^2} + \frac{\sum_{i \in C_k}{y_i}}{\sigma_k^2}\right)$ and $s_k^2 = \left(\frac{1}{\sigma_0^2} + \frac{n_k}{\sigma_k^2}\right)^{-1}$. 
\end{case}
Finally, hyper-parameters are specified as, $\mu_0=0$; $\sigma_0^2=100$; $\alpha_0=\beta_0=1/2$; $\alpha=4$. 

Here $50$ datasets are simulated from a Gaussian mixture with three components such that $\vect{\mu} = (-5,0,5)$ and $\vect{\sigma^2} = c(1,1,1)$. The WBIC and the power posterior 
approximations of the log-evidence are compared here; each power posterior estimate has $t_j = (j/(N))^5$ for $j = 1,2,\ldots,N=40$, as suggested by \shortcite{FrielHurnWyse2013}. 
Figure \ref{fig:mix1} presents the WBIC against the power posterior estimates of the evidence. Again there exists a tendency for WBIC to overestimate the evidence, relative to power posteriors, 
as has been exemplified for all four models under consideration.

\begin{figure}
	\centering
	\begin{subfigure}[t]{0.5\textwidth}
		\includegraphics[width=\textwidth]{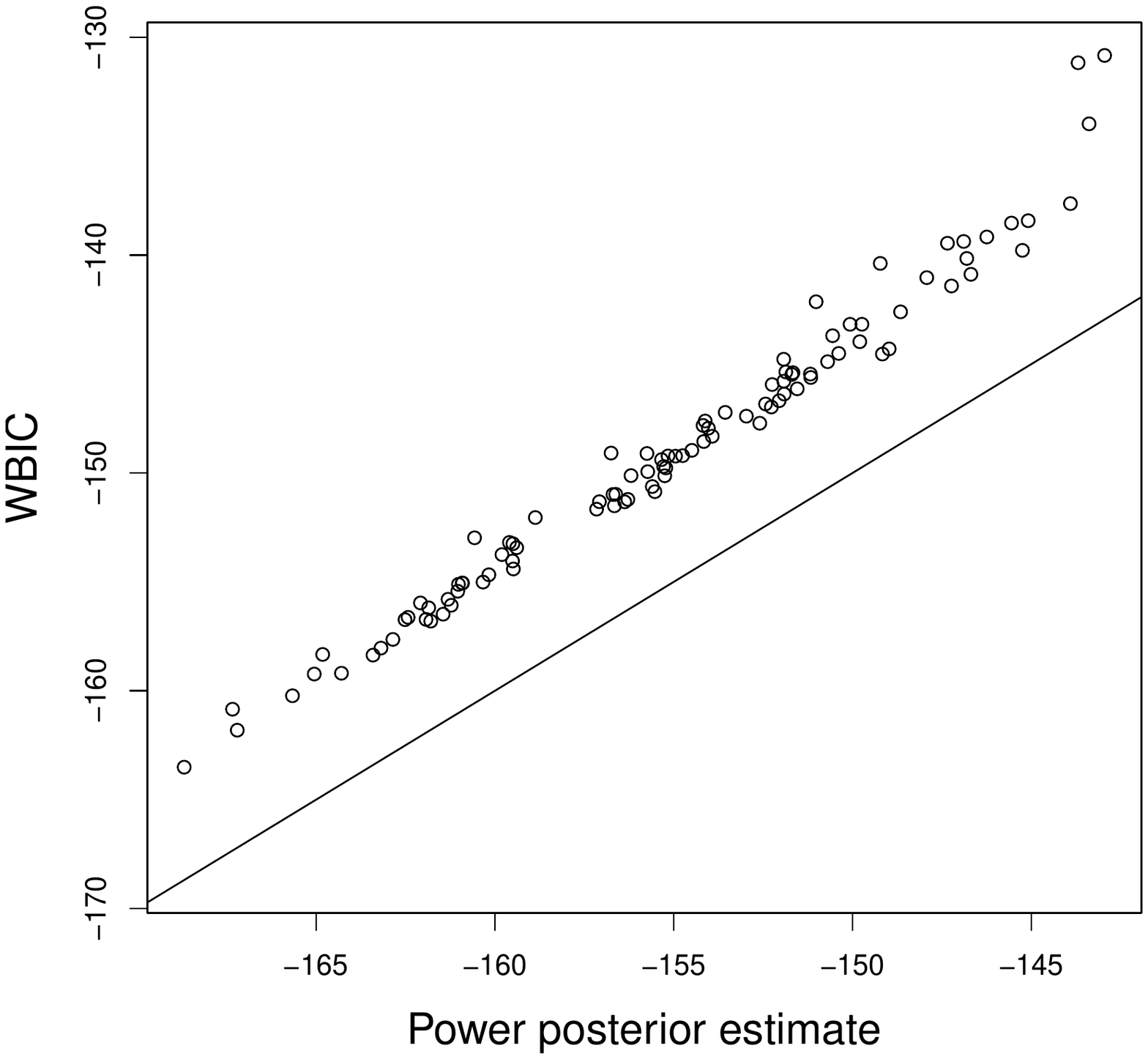}
  \end{subfigure}%
	~
	\begin{subfigure}[t]{0.5\textwidth}
		\includegraphics[width=\textwidth]{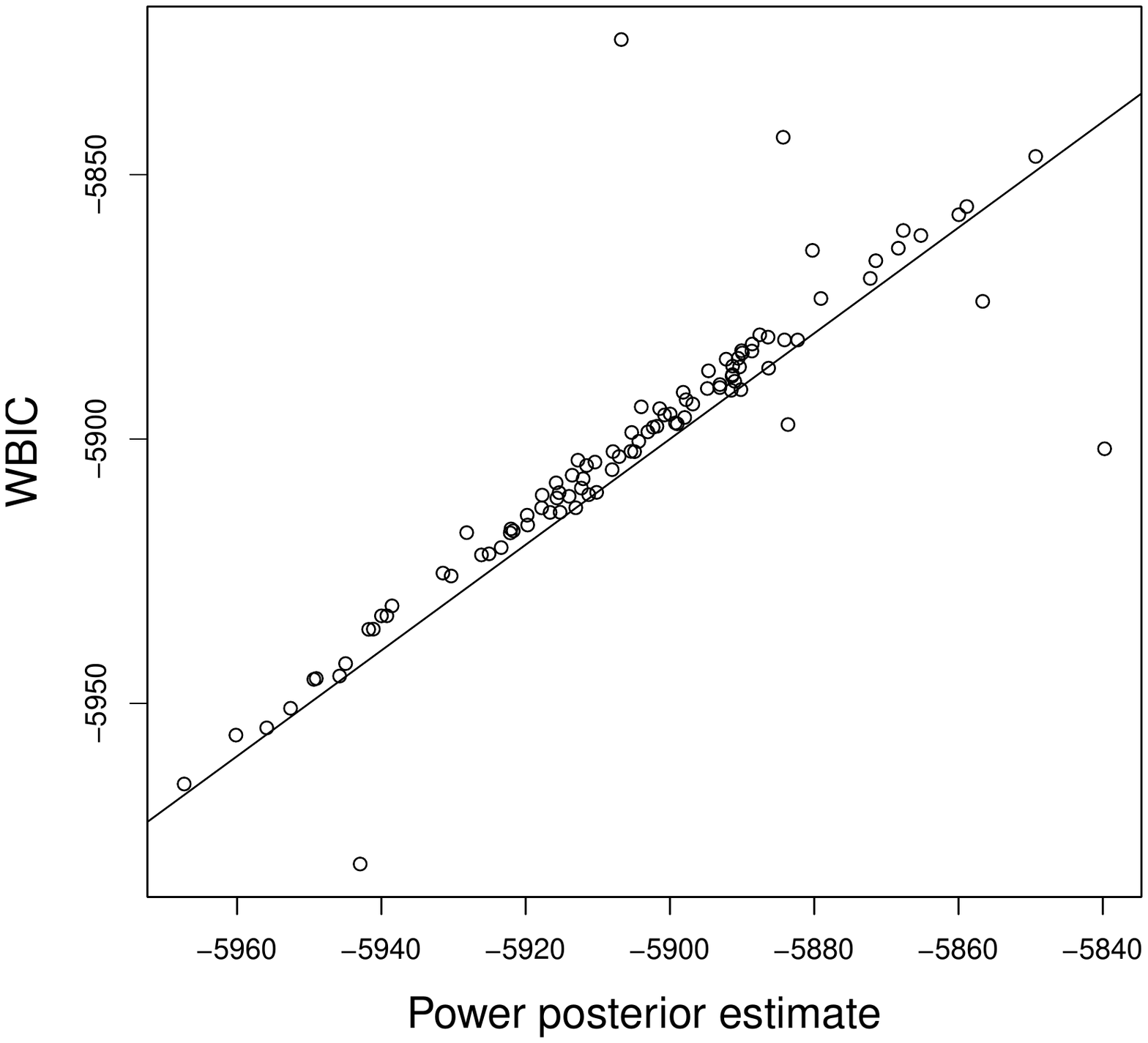}
  \end{subfigure}%
	\caption{Finite Gaussian mixture model: WBIC against the power posterior estimate of the evidence. \textbf{(a)} Sample of size $n=50$. \textbf{(b)} Sample of size $n=1000$. The 
	approximation is performing in a materially similar matter to the tractable normal model.}
	\label{fig:mix1}
\end{figure}

As before the analysis was repeated under a unit information prior formulation. Results, not presented here, were very similar to those for the tractable normal model.

\section{Discussion}
\label{Discussion}

Estimating the model evidence is well understood to be a very challenging task. Watanabe's WBIC is an interesting contribution to this literature. Although motivated from 
statistical learning theory, in principle it can be applied to both regular and singular models. From an implementational point of view, it offers to provide a computationally
cheap approximation of the evidence and this is an overwhelming advantage in favour of its use. 
Our theoretical case-study has suggested that an optimally-tuned WBIC estimator (where one has access to the optimal temperature $t^*$) is likely to perform better than the 
power posterior approach. 
However, the empirical experiments in this paper suggest that it can provide a poor estimate
of the evidence in practice, when $t_w$ is substituted for $t^*$, particularly in cases where one uses weakly informative priors. Of course, it has been argued in the literature that specification of 
priors for statistical model selection requires careful choice. In particular, unit information priors are often advocated for this purpose. Our study suggests that WBIC could provide a useful 
and cost-effective approach in this case. 

In terms of future directions, an interesting question to investigate would be whether one could improve upon the default temperature, $t_w = 1/\log(n)$. Insights such as Lemma~\ref{lem:kl}
may provide a useful starting point and we are currently working in this direction. 

\subsubsection*{Acknowledgements}
The Insight Centre for Data Analytics is supported by Science Foundation Ireland under Grant Number SFI/12/RC/2289. Nial Friel's research was also supported by an Science Foundation Ireland grant: 12/IP/1424. 
James McKeone is grateful for the support of an Australian Postgraduate Award (APA). Tony Pettitt's research was supported by the Australian Research Council Discovery Grant, DP$1101000159$.

\bibliography{wbic}
\end{document}